\documentclass[11pt]{article}

\usepackage[a4paper,margin=1in]{geometry}
\usepackage{amsmath,amssymb,amsthm}
\usepackage{authblk}
\usepackage{hyperref}
\usepackage{enumitem}
\usepackage{mathtools}

\hypersetup{colorlinks=true,linkcolor=blue,citecolor=blue,urlcolor=blue}

\newtheorem{theorem}{Theorem}
\newtheorem{proposition}[theorem]{Proposition}
\newtheorem{lemma}[theorem]{Lemma}
\newtheorem{corollary}[theorem]{Corollary}
\theoremstyle{definition}
\newtheorem{definition}[theorem]{Definition}
\newtheorem{remark}[theorem]{Remark}

\title{On the Simulation Cost of Quantum Finite Automata}
\author[1]{Zeyu Chen\thanks{\texttt{chenzeyu@zju.edu.cn}}}
\author[1]{Junde Wu\thanks{\texttt{wjd@zju.edu.cn}}}
\affil[1]{School of Mathematical Sciences, Zhejiang University, Hangzhou 310058, People's Republic of China}
\date{}

\begin{document}
\maketitle

\begin{abstract}
This paper identifies exact probabilistic simulation cost as the natural quantitative measure of quantum advantage for finite automata under strict cutpoints. It gives sharp simulation laws for two representative models.  A one-way finite automaton with $c$ classical states and a $q$-dimensional quantum register has exact probabilistic simulation cost $\Theta(cq^2)$, while an $n$-dimensional measure-once one-way quantum finite automaton has worst-case cost $\Theta(n^2)$.  The proofs develop a prepare--test framework, in which prefixes generate the relevant real operator degrees of freedom and suffixes convert them into strict-cutpoint tests.  The same obstruction is recast through finite sign-rank matrices, clarifying the role of Forster's spectral method.  Placed beside the surrounding two-way separations, these results give a clean hierarchy of finite-automata quantum advantage.
\end{abstract}

\section{Introduction}

Finite automata provide a precise laboratory for separating forms of quantum and classical computation.  Their memory is finite and their transition rules are explicit, yet small changes in the model---unitary evolution, intermediate measurements, general quantum operations, hybrid quantum--classical control---can change the recognized languages, the number of states needed, or the cost of classical simulation.  The resulting landscape has been developed through a sequence of foundational models~\cite{MooreCrutchfield2000,KondacsWatrous1997,AmbainisFreivalds1998,AmbainisWatrous}; broader surveys may be found in~\cite{BertoniMereghettiPalano2003,AmbainisYakaryilmaz2021}.

Quantum advantage in this setting takes three progressively finer forms.  \emph{Recognition power} asks whether a quantum model recognizes languages beyond a corresponding probabilistic model.  \emph{State succinctness} asks how much larger the classical description must be when both models can recognize the same language.  \emph{Exact simulation cost} refines the comparison further: once the chosen semantics makes probabilistic recognition possible, how many probabilistic states are needed to reproduce the same threshold language?  This paper treats simulation cost as the central quantitative invariant, with recognition and succinctness providing the surrounding context.

Strict cutpoints supply the common semantics.  Probabilistic finite automata (PFA) recognize exactly the stochastic languages under strict cutpoints~\cite{Rabin1963,Paz1971,Turakainen1969}, and the one-way quantum models considered below linearize to generalized finite automata and therefore remain in the same recognition regime~\cite{YakaryilmazSay2009CSR,LiQiuZouEtAl2012,AmbainisYakaryilmaz2021}.  Within this effective one-way setting, the quantum and probabilistic models have the same qualitative language power, and the meaningful comparison becomes: given a quantum automaton with acceptance function~$f$ and cutpoint~$\lambda$, how many PFA states are needed to produce the threshold language $\{w:f(w)>\lambda\}$?

The method developed here is \emph{prepare--test}.  A short prefix prepares an internal configuration; a short suffix tests it by moving the acceptance probability across the cutpoint.  When many configurations can be tested independently, any probabilistic simulator must realize the same family of threshold tests by affine halfspaces on its probability simplex.  The lower bounds are therefore governed by the real degrees of freedom visible to the accepting functional, while matching upper bounds come from linearizing the quantum dynamics into a generalized finite automaton and converting it to a PFA.

Two sharp one-way simulation laws result.  For one-way finite automata with quantum and classical states (1QCFA), the governing dimension is that of the block-diagonal classical--quantum operator space, namely $cq^2$, yielding exact simulation cost $\Theta(cq^2)$.  For measure-once one-way quantum finite automata (MO-1QFA), the reachable pure-state manifold has only linear dimension, but the accepting measurement compensates: a balanced projector moves on a unitary orbit of dimension $\lfloor n^2/2\rfloor$, and together with the mixed-state upper bound from~\cite{ChenWuQuadratic}, this gives worst-case cost $\Theta(n^2)$.  In both cases, upper and lower bounds match because the argument identifies the hidden geometric dimension that a probabilistic simulator must reproduce.

Finite prefix--suffix restrictions cast the same obstruction in sign-rank form.  Fixing finite prefix and suffix sets, a quantum automaton determines a sign matrix recording which concatenations lie above the cutpoint.  Since PFA state distributions lie in a probability simplex, an $m$-state simulator gives a rank-$m$ real realization of that sign pattern.  The prepare--test witnesses therefore have an exact sign-rank formulation, and their comparison with Forster's spectral method reveals a useful distinction: the basic spectral certificate operates at square-root scale, while the complete shattering witness used here records the full threshold geometry~\cite{Forster2002,PaturiSimon1986,LinialMendelsonSchechtmanShraibman2007}.

The broader landscape completes the picture.  Bounded-error two-way quantum--classical automata (2QCFA) already separate from classical models through recognition power and time complexity~\cite{AmbainisWatrous}.  These two-way separations delineate the effective one-way regime studied here: before two-way motion enters, strict-cutpoint recognition is shared with probabilistic automata, and quantum advantage is measured by exact simulation cost.

The paper is organized as follows.  Section~\ref{sec:prelim} fixes notation and introduces the model and tool conventions.  Section~\ref{sec:prepare-test} proves the prepare--test bounds for 1QCFA and MO-1QFA\@.  Section~\ref{sec:signrank} recasts those bounds through sign-rank and gives a focused account of Forster's spectral method.  Section~\ref{sec:twoway-boundary} places the results in the landscape of quantum advantage and identifies the boundary of simulation.  Section~\ref{sec:discussion} concludes.

\section{Preliminaries}
\label{sec:prelim}

This section fixes notation and collects the classical, quantum, geometric, and combinatorial tools needed for the proofs.  Generalized and probabilistic finite automata supply the linear framework for upper bounds; Hermitian-operator coordinates and Grassmannian geometry supply the lower-bound machinery; VC dimension and sign-rank translate geometric shattering into state-count lower bounds.  The broader comparison between quantum advantage, state succinctness, and probabilistic simulation is deferred to Section~\ref{sec:twoway-boundary}, after the one-way simulation bounds and their sign-rank form have been established.

\subsection{Notation and basic conventions}

The symbols $\mathbb{R}$ and $\mathbb{C}$ denote the sets of real and complex numbers, respectively, and $\mathbb{R}_{\ge 0}$ denotes the set of nonnegative real numbers. For a positive integer $m$, set $[m]=\{1,\dots,m\}$.
For $q\ge 1$, the notation $\mathbb{CP}^{q-1}$ denotes the complex projective space of one-dimensional complex subspaces of $\mathbb{C}^q$, equivalently the set of unit vectors in $\mathbb{C}^q$ modulo global phase.
If $X$ and $Y$ are sets, then $X\times Y=\{(x,y): x\in X,\ y\in Y\}$ denotes their Cartesian product.

Let $\Sigma$ be a finite alphabet and let $\Sigma^\ast$ denote the free monoid generated by $\Sigma$. For a word $w\in\Sigma^\ast$, its length is denoted $|w|$.
All real vectors are written in row-vector convention when probabilistic automata are discussed.
For $m\ge 1$, the probability simplex is
\[
\Delta_{m-1}=\bigl\{x\in\mathbb{R}_{\ge 0}^m: \sum_{i=1}^m x_i=1\bigr\}.
\]

\begin{definition}
\label{def:strict-cutpoint}
If a machine $A$ has acceptance-probability function $f_A:\Sigma^\ast\to[0,1]$ and $\lambda\in\mathbb{R}$, then $A$ \emph{recognizes a language $L$ with strict cutpoint $\lambda$} if
\[
L=\{w\in\Sigma^\ast: f_A(w)>\lambda\}.
\]
\end{definition}
This paper considers strict cutpoints throughout.
The languages recognized by probabilistic finite automata under this semantics are the \emph{stochastic languages}~\cite{Rabin1963,Paz1971,Turakainen1969}.

For comparison with bounded-error state complexity, the standard isolated-cutpoint terminology will also be used.  A machine with acceptance probability $f_A$ recognizes $L$ with cutpoint $\lambda$ and isolation margin $\eta>0$ if
\[
w\in L \Rightarrow f_A(w)\ge \lambda+\eta,
\qquad
w\notin L \Rightarrow f_A(w)\le \lambda-\eta.
\]
The technical results below use strict cutpoints; isolated cutpoints appear only in the contextual discussion of state succinctness in Section~\ref{sec:twoway-boundary}.

\subsection{Generalized and probabilistic finite automata}

The linear model that mediates all upper bounds in this paper is the generalized finite automaton (GFA).
Following the standard effective convention for probabilistic automata, ordinary PFA and 2PFA are rational unless a real-probability variant is explicitly specified~\cite{Rabin1963,Paz1971,DworkStockmeyer1990}.
The cutpoint is part of the recognition semantics and will be specified separately.
For an ordered subfield $K\subseteq\mathbb{R}$, a $K$-PFA means the same finite-state probabilistic model with all stochastic data in $K$; thus an ordinary PFA in this paper is a $\mathbb{Q}$-PFA.

\begin{definition}
\label{def:gfa}
A \emph{generalized finite automaton} (GFA) over $\Sigma$ is a tuple
\[
G=(S,\Sigma,u,\{A_\sigma\}_{\sigma\in\Sigma},v),
\]
where $|S|=k$, $u\in\mathbb{R}^{1\times k}$ is the initial row vector, each $A_\sigma\in\mathbb{R}^{k\times k}$ is a transition matrix, and $v\in\mathbb{R}^{k\times 1}$ is the final column vector.
For $w=\sigma_1\cdots\sigma_m$, its acceptance value is
\[
f_G(w)=uA_{\sigma_1}\cdots A_{\sigma_m}v.
\]
\end{definition}

\begin{definition}
\label{def:pfa}
A \emph{one-way probabilistic finite automaton} (PFA) over $\Sigma$ is a tuple
\[
P=(S,\Sigma,\pi,\{P_\sigma\}_{\sigma\in\Sigma},P_\#,F),
\]
where $\pi\in\Delta_{|S|-1}\cap\mathbb{Q}^{|S|}$ is the initial distribution, each $P_\sigma$ and $P_\#$ is row-stochastic with rational entries, and $F\subseteq S$ is the accepting set.
Writing $\mathbf{1}_F$ for the column indicator vector of $F$, the acceptance probability of
$w=\sigma_1\cdots\sigma_m$ is
\[
f_P(w)=\pi P_{\sigma_1}\cdots P_{\sigma_m}P_\#\mathbf{1}_F.
\]
\end{definition}

The quantitative conversion from GFAs to probabilistic automata used in the upper bounds is recorded in the following proposition.
It may be viewed as a quantitative form of the strict-cutpoint GFA--PFA correspondence: the underlying language-theoretic idea goes back to~\cite{Turakainen1969}, and the alphabet-preserving version with the explicit state count below is taken from~\cite{ChenWuQuadratic}.

\begin{proposition}
\label{prop:gfa-to-pfa}
Let $K\subseteq\mathbb{R}$ be an ordered subfield.
For every $k$-state GFA $G$ over $\Sigma$ whose entries lie in $K$ and every strict cutpoint $\lambda\in K$, there exists an alphabet-preserving one-way $K$-PFA with at most $2k+6$ states and strict cutpoint $1/2$ recognizing the same language.
In particular, rational GFA data give an ordinary PFA in the convention above.
\end{proposition}

\subsection{Quantum notation}
\label{subsec:quantum-notation}

The paper uses standard notation from finite-dimensional quantum information; general background on quantum states, measurements, and channels may be found in~\cite{Watrous2018,BengtssonZyczkowski2017}.
Let $\mathcal{H}\cong\mathbb{C}^q$ be a finite-dimensional Hilbert space, let $\mathcal{L}(\mathcal{H})$ denote the space of all linear operators on $\mathcal{H}$, and let $\mathrm{Herm}(\mathcal{H})$ denote the real vector space of Hermitian operators.
Then $\dim_{\mathbb{R}} \mathrm{Herm}(\mathcal{H})=q^2$, and its traceless subspace has dimension $q^2-1$.
The Hilbert--Schmidt inner product on operators is
\[
\langle A,B\rangle_{\mathrm{HS}}=\mathrm{Tr}(A^\dagger B).
\]
A \emph{density operator} is a positive semidefinite operator $\rho\in\mathcal{L}(\mathcal{H})$ with $\mathrm{Tr}(\rho)=1$.
The density operators form a compact subset of the affine hyperplane $\{X\in \mathrm{Herm}(\mathcal{H}) : \mathrm{Tr}(X)=1\}$.
A \emph{pure state} is a rank-one density operator $|\psi\rangle\langle\psi|$ for a unit vector $|\psi\rangle\in\mathcal{H}$; equivalently, pure states are the points of the complex projective space $\mathbb{CP}^{q-1}$, because unit vectors that differ only by a global phase define the same rank-one projector.
Following Ludwig~\cite{Ludwig1964}, an \emph{effect} is an operator $E$ satisfying $0\preceq E\preceq I$, and a \emph{projector} is an effect satisfying $P^2=P=P^\dagger$.
A projective accept--reject measurement is therefore specified by a projector $P_{\mathrm{acc}}$ together with its complement $I-P_{\mathrm{acc}}$.
A \emph{quantum channel} on $\mathcal{L}(\mathcal{H})$ is a completely positive trace-preserving (CPTP) linear map.

For integers $0\le r\le n$, the \emph{complex Grassmannian} $\mathrm{Gr}(r,n)$ is the set of all $r$-dimensional complex linear subspaces of $\mathbb{C}^n$. Equivalently, each such subspace determines a unique orthogonal projector onto it, so $\mathrm{Gr}(r,n)$ can be identified with the set of rank-$r$ orthogonal projectors on $\mathbb{C}^n$.
This is exactly the form needed later: if $P_0$ is a fixed rank-$r$ projector, then every conjugate $U^\dagger P_0U$ has the same rank, and every rank-$r$ projector arises in this way.
Hence
\[
\mathrm{Gr}(r,n) = \{U^\dagger P_0U:U\in U(n)\}.
\]
Two geometric facts about this space are used later: it is a compact smooth manifold, and its real dimension is $2r(n-r)$.
Standard background for the differential-geometric notions used below, including tangent spaces, differentials, and the inverse function theorem, may be found in~\cite{LeeSmoothManifolds}.

\subsection{Automata models}
\label{subsec:automata-models}

The hybrid one-way model has a finite classical controller and a bounded quantum register.

\begin{definition}
\label{def:1qcfa}
A \emph{one-way finite automaton with quantum and classical states}, abbreviated $(c,q)$-1QCFA, over a finite alphabet $\Sigma$ is a tuple
\[
\mathcal{A}=(S,\mathcal{H},\Sigma,s_0,\rho_0,\delta,\Phi,P_{\mathrm{acc}}),
\]
where:
\begin{itemize}[nosep]
\item $S$ is a finite set of classical states with $|S|=c$;
\item $\mathcal{H}\cong\mathbb{C}^q$ is the Hilbert space of the quantum register;
\item $s_0\in S$ is the initial classical state;
\item $\rho_0$ is an initial density operator on $\mathcal{H}$;
\item $\delta:S\times\Sigma\to S$ is the deterministic classical transition function;
\item $\Phi=\{\Phi_{s,\sigma}:s\in S,\ \sigma\in\Sigma\}$ is a family of quantum channels on $\mathcal{L}(\mathcal{H})$, where each $\Phi_{s,\sigma}$ is completely positive and trace preserving;
\item $P_{\mathrm{acc}}$ is an accepting projector on $\mathcal{H}$.
\end{itemize}
For an input word $w=\sigma_1\cdots\sigma_n$, the computation is defined recursively by
\[
s_i=\delta(s_{i-1},\sigma_i),
\qquad
\rho_i=\Phi_{s_{i-1},\sigma_i}(\rho_{i-1}),
\qquad 1\le i\le n.
\]
After the input has been read, the acceptance probability is
\[
f_{\mathcal{A}}(w)=\mathrm{Tr}\!\left(P_{\mathrm{acc}}\rho_n\right).
\]
\end{definition}

\begin{remark}
The 1QCFA results below remain valid for the effect-valued variant in the sense of Ludwig~\cite{Ludwig1964}, where the final accepting projector is replaced by an arbitrary accepting effect $E_{\mathrm{acc}}$ satisfying $0\preceq E_{\mathrm{acc}}\preceq I$.
\end{remark}

The present 1QCFA convention extends the one-way quantum finite automata with classical states studied in~\cite{ZhengQiuGruska2013,QiuLiMateusSernadas2015} by allowing mixed quantum states and general quantum operations.
After reading $w=\sigma_1\cdots\sigma_n$, the actual classical--quantum configuration of $\mathcal{A}$ is the one-block operator
\[
|s_n\rangle\langle s_n|\otimes\rho_n.
\]
The following block notation will also be used
\[
X_{CQ}=\sum_{s\in S}|s\rangle\langle s|\otimes X_s,
\qquad X_s\in\mathrm{Herm}(\mathcal{H}).
\]
A positive trace-one block operator has the form
\[
\rho_{CQ}=\sum_{s\in S}|s\rangle\langle s|\otimes\rho_s,
\]
where each $\rho_s$ is positive semidefinite and $\sum_{s\in S}\mathrm{Tr}(\rho_s)=1$.
The symbol $\sigma\in\Sigma$ induces the block map
\[
T_\sigma(\rho_{CQ})
=
\sum_{t\in S}|t\rangle\langle t|\otimes
\left(
\sum_{s:\delta(s,\sigma)=t}\Phi_{s,\sigma}(\rho_s)
\right),
\]
and the acceptance functional on such block-diagonal operators is
\[
F_{\mathrm{acc}}(\rho_{CQ})
=\sum_{s\in S}\mathrm{Tr}\!\left(P_{\mathrm{acc}}\rho_s\right).
\]
On one-block configurations $|s\rangle\langle s|\otimes\rho$, this notation agrees with Definition~\ref{def:1qcfa}.

The pure-state measure-once model is recalled next.

\begin{definition}
\label{def:mo1qfa}
An \emph{$n$-dimensional pure-state measure-once one-way quantum finite automaton}, abbreviated $n$-MO-1QFA, is a tuple
\[
\mathcal{Q}=(\mathcal{H},\Sigma,|\psi_0\rangle,\{U_\sigma\}_{\sigma\in\Sigma},P_{\mathrm{acc}}),
\]
where $\dim(\mathcal{H})=n$, the vector $|\psi_0\rangle$ is the initial unit state, each $U_\sigma$ is unitary on $\mathcal{H}$, and $P_{\mathrm{acc}}$ is a projector.
For $w=\sigma_1\cdots\sigma_m$, the final state is $|\psi_w\rangle=U_{\sigma_m}\cdots U_{\sigma_1}|\psi_0\rangle$, and the acceptance probability is
\[
f_{\mathcal{Q}}(w)=\langle\psi_w|P_{\mathrm{acc}}|\psi_w\rangle.
\]
\end{definition}

This is the original one-way pure-state model of~\cite{MooreCrutchfield2000}; further background in the language-theoretic direction may be found in~\cite{BrodskyPippenger2002}.

For the contextual discussion in Section~\ref{sec:twoway-boundary}, a \emph{two-way probabilistic finite automaton}, abbreviated 2PFA, is the standard two-way version of a PFA: it has a finite probabilistic controller, one read-only head moving on an endmarked input tape, and designated accepting and rejecting halting states.  A \emph{two-way finite automaton with quantum and classical states}, abbreviated 2QCFA, is used in the standard sense of Ambainis and Watrous~\cite{AmbainisWatrous}: a finite classical controller moves a two-way input head while applying finite-dimensional quantum operations to a constant-size quantum register.

\subsection{VC dimension and sign-rank}
\label{subsec:vc-signrank}

Two combinatorial notions are used in the lower bounds.

\begin{definition}
\label{def:vc}
If $X$ is a set and $\mathcal{C}\subseteq 2^X$ is a family of subsets, a finite set $Y\subseteq X$ is \emph{shattered} by $\mathcal{C}$ if for every $Z\subseteq Y$ there exists $C\in\mathcal{C}$ with $C\cap Y=Z$.
The \emph{Vapnik--Chervonenkis (VC) dimension} $\mathrm{VCdim}(\mathcal{C})$ is the maximum cardinality of a shattered finite subset of $X$.
\end{definition}
The following standard VC-dimension fact for affine halfspaces will be used; see, for example,~\cite[Chapter~4]{AnthonyBartlett} and~\cite[Chapter~10]{VapnikBook}.

\begin{proposition}
\label{prop:simplex-vc}
Let
\[
\mathcal{H}_m=\bigl\{\{x\in\Delta_{m-1}: a^\top x>\theta\}: a\in\mathbb{R}^m,\ \theta\in\mathbb{R}\bigr\}.
\]
Then $\mathrm{VCdim}(\mathcal{H}_m)=m$.
\end{proposition}

\begin{definition}
\label{def:signrank}
If $X$ and $Y$ are finite sets, a \emph{sign matrix} is a matrix $S\in\{\pm 1\}^{X\times Y}$.
Its \emph{sign-rank}, denoted $\mathrm{rank}_{\pm}(S)$, is the minimum rank of a real matrix $R$ such that $S_{x,y}R_{x,y}>0$ for every $(x,y)\in X\times Y$.
\end{definition}

\section{Prepare--test bounds for one-way quantum models}
\label{sec:prepare-test}

Both one-way lower bounds in this paper rest on the same two-step mechanism.  First, a short prefix \emph{prepares} a chosen internal configuration of the quantum automaton.  Second, a short suffix \emph{tests} that configuration by driving its acceptance probability above or below the strict cutpoint.  Any probabilistic simulator must then realize the same family of threshold decisions by affine halfspaces on its simplex of state distributions, forcing a dimension lower bound on the simulator.  The two subsections below implement this mechanism for 1QCFA and for measure-once pure-state automata, revealing two distinct geometric sources of the quadratic cost.

\subsection{1QCFA and the \texorpdfstring{$\Theta(cq^2)$}{Theta(cq2)} law}
\label{sec:hybrid}

This subsection establishes that the correct strict-cutpoint simulation invariant for $(c,q)$-1QCFA is the real dimension of the block-diagonal classical--quantum operator space.
The proof first gives a direct linearization of the 1QCFA dynamics and then shows that this dimension is genuinely saturated by strict-cutpoint distinguishability.

\begin{theorem}
\label{thm:hybrid-main}
For every $c\ge 2$ and $q\ge 2$, the following upper and lower bounds hold for strict-cutpoint simulation of $(c,q)$-1QCFA:
\begin{enumerate}[label=(\alph*),nosep]
\item every $(c,q)$-1QCFA with strict cutpoint $\lambda$ admits exact simulation by a one-way $K$-PFA with at most $2cq^2+6$ states, where $K$ is any ordered subfield containing $\lambda$ and the entries of the linearized GFA in Theorem~\ref{thm:cq-linearization};
\item there exists a $(c,q)$-1QCFA over a finite alphabet and strict cutpoint $1/2$ for which every equivalent one-way PFA has at least $cq^2-1$ states.
\end{enumerate}
\end{theorem}

The two assertions are proved in order below.  We begin with the upper bound, which is obtained by linearizing the block-diagonal classical--quantum dynamics.  The key observation is that the ambient real operator space has dimension $cq^2$, because the classical controller contributes a direct-sum index and leaves one $q^2$-dimensional quantum operator space in each classical block.

\begin{theorem}
\label{thm:cq-linearization}
Let $\mathcal{A}$ be a $(c,q)$-1QCFA\@.
There exists a $cq^2$-state GFA $G_{\mathcal{A}}$ such that
\[
 f_{G_{\mathcal{A}}}(w)=f_{\mathcal{A}}(w)
\quad\text{for all } w\in\Sigma^\ast.
\]
\end{theorem}

\begin{proof}
Fix an ordering $S=\{1,\dots,c\}$, and let $\mathcal{H}$ be the quantum register of $\mathcal{A}$.
Fix an orthonormal basis $B_0,\dots,B_{q^2-1}$ of $\mathrm{Herm}(\mathcal{H})$ with respect to the Hilbert--Schmidt inner product $\langle\cdot,\cdot\rangle_{\mathrm{HS}}$.
For each block $\rho_i$, define its coordinate vector
\[
 x(\rho_i)=\bigl(\mathrm{Tr}(B_0\rho_i),\dots,\mathrm{Tr}(B_{q^2-1}\rho_i)\bigr)^\top\in\mathbb{R}^{q^2}.
\]
Concatenating these block coordinates gives the state vector of the linearized system:
\[
 x(\rho_{CQ})=\bigl(x(\rho_1)^\top,\dots,x(\rho_c)^\top\bigr)^\top\in\mathbb{R}^{cq^2}.
\]
For a fixed symbol $\sigma$, the block transition $T_\sigma$ associated with Definition~\ref{def:1qcfa} is linear on the real vector space of block Hermitian operators.
Hence there exists a real matrix $M_\sigma\in\mathbb{R}^{cq^2\times cq^2}$ such that
\[
 x(T_\sigma(\rho_{CQ}))=M_\sigma\,x(\rho_{CQ}).
\]

The acceptance functional is likewise linear in the final classical--quantum state.
Indeed, if
\[
\rho_{CQ}=\sum_{i=1}^c |i\rangle\langle i|\otimes \rho_i,
\]
then the 1QCFA accepts with probability
\[
f_{\mathcal{A}}(\rho_{CQ})=\sum_{i=1}^c \mathrm{Tr}(P_{\mathrm{acc}}\rho_i).
\]
The next step is to write this blockwise trace functional in the coordinates fixed above.
Define
\[
 v_{\mathrm{acc}}=\bigl(\mathrm{Tr}(P_{\mathrm{acc}}B_0),\dots,\mathrm{Tr}(P_{\mathrm{acc}}B_{q^2-1})\bigr)^\top\in\mathbb{R}^{q^2},
\qquad
 v_{\mathrm{tot}}=\bigl(v_{\mathrm{acc}}^\top,\dots,v_{\mathrm{acc}}^\top\bigr)^\top\in\mathbb{R}^{cq^2}.
\]
For each block $\rho_i$, the equality
\[
\mathrm{Tr}(P_{\mathrm{acc}}\rho_i)=v_{\mathrm{acc}}^\top x(\rho_i)
\]
follows from the definition of the coordinate vector $x(\rho_i)$.
Summing over the $c$ blocks gives
\[
f_{\mathcal{A}}(w)=v_{\mathrm{tot}}^\top x(\rho_{CQ}(w)).
\]
Here $\rho_{CQ}(0)=|s_0\rangle\langle s_0|\otimes\rho_0$ is the initial block state.
Define the GFA $G_{\mathcal{A}}$ with initial row vector $u^\top=x(\rho_{CQ}(0))^\top$, transition matrices $A_\sigma=M_\sigma^\top$, and final column vector $v_{\mathrm{tot}}$.
A direct computation confirms $f_{G_{\mathcal{A}}}(w)=u A_w v_{\mathrm{tot}}=f_{\mathcal{A}}(w)$.
\end{proof}

\begin{proof}[Proof of Theorem~\ref{thm:hybrid-main}(a)]
Theorem~\ref{thm:cq-linearization} gives a GFA with $cq^2$ states and the same acceptance probability function.
Applying Proposition~\ref{prop:gfa-to-pfa} over any ordered subfield containing the linearized GFA entries and the cutpoint converts this GFA into a one-way $K$-PFA with at most $2cq^2+6$ states recognizing the same strict-cutpoint language.
\end{proof}

After the upper bound has identified $cq^2$ as the ambient dimension, the lower-bound half shows that this dimension is genuinely saturated by strict-cutpoint distinguishability.  The argument follows the prepare--test philosophy from the mixed-state setting~\cite{ChenWuQuadratic}, now adapted to deterministic classical control and block-diagonal effects.

The first ingredient is a channel that realizes any prescribed effect through a fixed projective measurement.

\begin{lemma}
\label{lem:effect-realization}
Let $\dim(\mathcal{H})=q\ge 2$, and let $E$ be an effect on $\mathcal{H}$.
Fix an orthonormal basis $\{|1\rangle,\dots,|q\rangle\}$ and let $P_{\mathrm{acc}}=|1\rangle\langle 1|$.
There exists a CPTP map $\Phi_E$ such that
\[
 \mathrm{Tr}(P_{\mathrm{acc}}\,\Phi_E(\rho))=\mathrm{Tr}(E\rho)
\quad\text{for all density operators }\rho.
\]
\end{lemma}

\begin{proof}
Write the spectral decomposition $E=\sum_{i=1}^q \lambda_i |\psi_i\rangle\langle\psi_i|$ with $0\le \lambda_i\le 1$.
Define Kraus operators
\[
K_i=\sqrt{\lambda_i}\,|1\rangle\langle\psi_i|,
\qquad
L_i=\sqrt{1-\lambda_i}\,|2\rangle\langle\psi_i|.
\]
Then $\sum_i (K_i^{\dagger}K_i+L_i^{\dagger}L_i)=I$, so the map $\Phi_E(\rho)=\sum_i K_i\rho K_i^{\dagger}+\sum_i L_i\rho L_i^{\dagger}$ is CPTP\@.
A direct computation gives $\mathrm{Tr}(P_{\mathrm{acc}}\Phi_E(\rho))=\sum_i \mathrm{Tr}(K_i^{\dagger}K_i\rho)=\mathrm{Tr}(E\rho)$.
\end{proof}

\begin{proof}[Proof of Theorem~\ref{thm:hybrid-main}(b)]
Let $\mathcal{K}\cong\mathbb{C}^c$ denote the classical register, with computational basis $\{|1\rangle,\dots,|c\rangle\}$.
Let $\mathcal{H}\cong\mathbb{C}^q$ denote the quantum register, with computational basis $\{|1\rangle,\dots,|q\rangle\}$.
The block-diagonal representation on $\mathcal{K}\otimes\mathcal{H}$ is the one fixed in Section~\ref{subsec:automata-models}.

The construction starts with deterministic classical--quantum configurations that fill the affine dimension of the block-diagonal state space.
Let $\omega=I_q/q$, and let $H_1,\dots,H_{q^2-1}$ be a basis of the traceless Hermitian operators on $\mathcal{H}$.
Choose $\varepsilon>0$ small enough that
\[
\theta_0=\omega,
\qquad
\theta_k=\omega+\varepsilon H_k,\quad 1\le k\le q^2-1,
\]
are density operators.
Such an $\varepsilon$ exists because $\omega$ is positive definite and each $H_k$ is traceless.
The operators $\theta_0,\dots,\theta_{q^2-1}$ are linearly independent in $\mathrm{Herm}(\mathcal{H})$.
For $1\le i\le c$ and $0\le k\le q^2-1$, set
\[
\zeta_{i,k}=|i\rangle\langle i|\otimes\theta_k.
\]
These $cq^2$ deterministic classical--quantum configurations are linearly independent in the block-diagonal Hermitian space, because different classical blocks have orthogonal support and the $\theta_k$ are linearly independent.
Select any
\[
d=cq^2-1
\]
of them and enumerate the selected configurations as
\[
\zeta_\ell=|a_\ell\rangle\langle a_\ell|\otimes\theta_{b_\ell},
\qquad 1\le \ell\le d.
\]
Since all $\zeta_\ell$ have trace one and are linearly independent, they are affinely independent.

These configurations are next separated by auxiliary block-diagonal effects indexed by sign patterns.
For each sign vector $\eta\in\{\pm1\}^d$, the linear independence of $\zeta_1,\dots,\zeta_d$ implies that there is a block-diagonal Hermitian operator $X_\eta$ on $\mathcal{K}\otimes\mathcal{H}$ such that
\[
\mathrm{Tr}(X_\eta\zeta_\ell)=\eta_\ell,
\qquad 1\le \ell\le d.
\]
Indeed, one may solve these equations inside the span of $\{\zeta_\ell:1\le \ell\le d\}$ using the invertible Hilbert--Schmidt Gram matrix.
Let
\[
M=\max_{\eta\in\{\pm1\}^d}\|X_\eta\|_{\mathrm{op}}.
\]
Choose $t>0$ such that $tM<1/2$ and define
\[
E_\eta=\frac12 I_{cq}+tX_\eta.
\]
Then $0\preceq E_\eta\preceq I_{cq}$, so $E_\eta$ is an effect on $\mathcal{K}\otimes\mathcal{H}$.
Since $E_\eta$ is block diagonal in the classical basis, it decomposes as
\[
E_\eta=\sum_{i=1}^c |i\rangle\langle i|\otimes E_\eta^{(i)},
\]
where each $E_\eta^{(i)}$ is an effect on $\mathcal{H}$ and
\[
\mathrm{Tr}(E_\eta\rho_{CQ})=\sum_{i=1}^c \mathrm{Tr}(E_\eta^{(i)}\rho_i)
\]
for every block-diagonal $\rho_{CQ}=\sum_i |i\rangle\langle i|\otimes\rho_i$.

These auxiliary effects can now be implemented inside a single 1QCFA by symbol-dependent channels.
Let the alphabet be
\[
\Sigma=\{p_1,\dots,p_d\}\cup\{\tau_\eta:\eta\in\{\pm1\}^d\}.
\]
Define a $(c,q)$-1QCFA $\mathcal{A}_{c,q}$ with classical state set $\{1,\dots,c\}$, quantum register $\mathcal{H}$, initial classical state $1$, initial quantum state $\theta_0$, and accepting projector $P_{\mathrm{acc}}=|1\rangle\langle1|$ on $\mathcal{H}$.

For the prepare symbol $p_\ell$, set $\delta(i,p_\ell)=a_\ell$ for every classical state $i$.
Let $\Phi_{i,p_\ell}=R_{\theta_{b_\ell}}$ for every $i$, where $R_{\theta_{b_\ell}}(\rho)=\mathrm{Tr}(\rho)\theta_{b_\ell}$ is the CPTP replacement channel.
Thus reading $p_\ell$ prepares the deterministic configuration $\zeta_\ell$ from any current configuration.

For the test symbol $\tau_\eta$, set $\delta(i,\tau_\eta)=i$ for every $i$.
Set $\Phi_{i,\tau_\eta}=\Phi_{E_\eta^{(i)}}$, where $\Phi_{E_\eta^{(i)}}$ is the channel supplied by Lemma~\ref{lem:effect-realization}.
The final accepting projector is the fixed projector $P_{\mathrm{acc}}$.

It remains to verify that this automaton realizes all strict-cutpoint labelings of the prepared configurations.
Evaluating the two-letter word $p_\ell\tau_\eta$ gives
\[
 f_{\mathcal{A}_{c,q}}(p_\ell\tau_\eta)
=\mathrm{Tr}\!\left(P_{\mathrm{acc}}\Phi_{E_\eta^{(a_\ell)}}(\theta_{b_\ell})\right)
=\mathrm{Tr}(E_\eta\zeta_\ell)
=\frac12+t\eta_\ell.
\]
Hence
\[
 f_{\mathcal{A}_{c,q}}(p_\ell\tau_\eta)>\frac12
\quad\Longleftrightarrow\quad \eta_\ell=+1.
\]
Thus the $d$ prepared configurations are shattered by strict cutpoint $1/2$.

Now suppose that an $m$-state PFA recognizes the same language under some strict cutpoint $\mu$.
Let $\delta_\ell\in\Delta_{m-1}$ be the state distribution after reading $p_\ell$, and for each $\tau_\eta$ let $b_\eta\in[0,1]^m$ be the vector of conditional acceptance probabilities after reading $\tau_\eta$ and the end-marker.
Then $f_P(p_\ell\tau_\eta)=\delta_\ell^\top b_\eta$.
The shattering property implies that the family of affine halfspaces on the simplex,
\[
\mathcal{C}_\mu=\{\,\{x\in\Delta_{m-1}: x^\top b>\mu\}\, :\, b\in[0,1]^m\},
\]
shatters the $d$ points $\delta_1,\dots,\delta_d$.
By Proposition~\ref{prop:simplex-vc}, affine halfspaces on $\Delta_{m-1}$ have VC dimension $m$. Therefore $d\le m$, and hence $m\ge cq^2-1$.
\end{proof}

The 1QCFA simulation cost is thus governed by a direct dimension count: the classical controller contributes a direct-sum index over $c$ blocks, each carrying a $q^2$-dimensional quantum operator space, and this product $cq^2$ is both sufficient and necessary.

\subsection{Pure-state measure-once automata and measurement-orbit shattering}
\label{sec:pure}

The pure-state measure-once model presents a different geometric challenge.  Pure states are parametrized by the complex projective space $\mathbb{CP}^{n-1}$, which has real dimension only $2n-2$, far smaller than the $n^2-1$ affine dimensions of the full mixed-state space~\cite{Watrous2018,BengtssonZyczkowski2017}.  If the reachable-state geometry alone determined the lower bound, MO-1QFA would admit a subquadratic simulation.  The quadratic obstruction instead comes from the accepting measurement: the unitary orbit of a balanced projector has local dimension $\lfloor n^2/2\rfloor$, and this freedom can be converted into a strict-cutpoint shattering construction by an explicit matrix-exponential argument.

The main result of this section is the following theorem.

\begin{theorem}
\label{thm:mo-quadratic}
For every $n\ge 2$, there exists an $n$-dimensional MO-1QFA under strict cutpoint $1/2$ such that every equivalent one-way PFA has at least $\lfloor n^2/2\rfloor$ states.
\end{theorem}

Theorem~\ref{thm:mo-quadratic} gives the lower-bound half of the pure-state simulation picture.
Combining it with the mixed-state simulation theorem from~\cite{ChenWuQuadratic} yields the corresponding worst-case statement immediately.

\begin{corollary}
\label{cor:mo-worstcase}
Every $n$-dimensional MO-1QFA with strict cutpoint $\lambda$ admits exact probabilistic simulation by a one-way $K$-PFA with at most $2n^2+6$ states, where $K$ is any ordered subfield containing $\lambda$ and the entries of the mixed-state linearization.
Moreover, there exists an $n$-dimensional MO-1QFA for which every equivalent one-way PFA has at least $\lfloor n^2/2\rfloor$ states.
\end{corollary}

The proof of Theorem~\ref{thm:mo-quadratic} starts from the local geometry of a balanced accepting projector and its unitary orbit, then turns that geometry into an explicit shattering construction.  A rank-$r$ accepting projector supplies exactly $d=\lfloor n^2/2\rfloor$ local deformation parameters, and carefully chosen pure states read off these parameters directly.  Since every chosen state has acceptance probability exactly $1/2$ at the balanced reference projector, a small unitary conjugation can move each test state independently to either side of the strict cutpoint.  The construction begins by fixing the balanced reference projector and the local coordinates on its unitary orbit.

Fix $n\ge 2$ and set
\[
 r=\lfloor n/2\rfloor,
 \qquad
 s=\lceil n/2\rceil,
 \qquad
 d=2rs=\lfloor n^2/2\rfloor.
\]
Decompose
\[
\mathcal{H}=\mathcal{H}_0\oplus \mathcal{H}_1,
\qquad \dim(\mathcal{H}_0)=r,
\qquad \dim(\mathcal{H}_1)=s,
\]
and let $P_0$ be the projector onto $\mathcal{H}_0$.
Its unitary orbit
\[
\mathcal{O}=\{U^\dagger P_0 U: U\in U(n)\}
\]
is precisely the set of rank-$r$ projectors on $\mathcal{H}$, hence the complex Grassmannian $\mathrm{Gr}(r,n)$.
This is a smooth manifold of real dimension $2r(n-r)=2rs=d$~\cite{LeeSmoothManifolds}.
The number $d$ is therefore not incidental: it is exactly the local dimension of the measurement orbit available for the construction.
The tangent directions at $P_0$ can be described in concrete block-matrix form.

\begin{lemma}
\label{lem:tangent}
A Hermitian matrix $\Delta$ lies in $T_{P_0}\mathcal{O}$ if and only if, relative to the decomposition $\mathcal{H}_0\oplus\mathcal{H}_1$,
\[
\Delta=\begin{pmatrix}
0 & X\\
X^\dagger & 0
\end{pmatrix}
\]
for some complex $r\times s$ matrix $X$.
In particular, $\dim_{\mathbb{R}}T_{P_0}\mathcal{O}=2rs=d$.
\end{lemma}

\begin{proof}
The orbit map $U\mapsto U^\dagger P_0U$ is smooth.
Differentiating at the identity along a skew-Hermitian direction $K$ gives $\Delta=[P_0,K]$.
Writing $K$ in block form and using $P_0=\mathrm{diag}(I_r,0)$ yields exactly the displayed description with an arbitrary off-diagonal block $X\in M_{r\times s}(\mathbb{C})$.
Conversely, every such $\Delta$ arises from a suitable choice of $K$.
\end{proof}

Lemma~\ref{lem:tangent} has a simple interpretation.
To first order, moving the projector does not come from unitary motion inside $\mathcal{H}_0$ or inside $\mathcal{H}_1$; those directions leave the range of $P_0$ unchanged.
Only the mixing between the two summands matters, and that mixing is encoded by the complex $r\times s$ matrix $X$.
Its real and imaginary parts provide the desired $2rs=d$ real coordinates.

To read off these $d$ real coordinates, fix the $d$ pure states
\[
|\psi_{a,b}^{R}\rangle=\frac{1}{\sqrt{2}}\bigl(|a\rangle+|r+b\rangle\bigr),
\qquad
|\psi_{a,b}^{I}\rangle=\frac{1}{\sqrt{2}}\bigl(|a\rangle+i|r+b\rangle\bigr),
\]
indexed by $(a,b)\in[r]\times[s]$.
Each satisfies $\langle\psi|P_0|\psi\rangle=1/2$.
These states are chosen because each one couples exactly one basis vector from $\mathcal{H}_0$ with one from $\mathcal{H}_1$.
At the reference projector $P_0$, they are perfectly balanced between acceptance and rejection.
Consequently, the first-order change in acceptance probability isolates the corresponding entry of the off-diagonal block $X$.
Define the evaluation map $\Phi:\mathcal{O}\to\mathbb{R}^{d}$ by collecting the corresponding acceptance probabilities:
\[
\Phi(E)=\Bigl(\langle \psi_{a,b}^{R}|E|\psi_{a,b}^{R}\rangle,\ \langle \psi_{a,b}^{I}|E|\psi_{a,b}^{I}\rangle\Bigr)_{(a,b)\in[r]\times[s]}.
\]

\begin{lemma}
\label{lem:jacobian}
The differential $d\Phi_{P_0}:T_{P_0}\mathcal{O}\to\mathbb{R}^{d}$ is an isomorphism.
Equivalently, the chosen test states extract the real part and the negative imaginary part of each off-diagonal entry independently.
\end{lemma}

\begin{proof}
Let $\Delta\in T_{P_0}\mathcal{O}$ and write it as in Lemma~\ref{lem:tangent} with off-diagonal block $X$.
Then, for each $(a,b)$,
\[
\langle \psi_{a,b}^{R}|\Delta|\psi_{a,b}^{R}\rangle=\mathrm{Re}(X_{a,b}),
\qquad
\langle \psi_{a,b}^{I}|\Delta|\psi_{a,b}^{I}\rangle=-\mathrm{Im}(X_{a,b}).
\]
Thus $d\Phi_{P_0}$ records all $2rs=d$ real coordinates of $X$, and is therefore bijective.
\end{proof}

Lemma~\ref{lem:jacobian} shows that, near $P_0$, the measurement orbit can be coordinatized by these acceptance probabilities.
In particular, prescribing the signs of the first-order changes in those probabilities is equivalent to prescribing the signs of all local coordinates.
The remaining task is therefore constructive: for each sign pattern, an explicit nearby conjugate projector is built whose acceptance probabilities move to the required side of $1/2$.  With these local coordinates in place, we can prove the theorem by an explicit shattering construction.

\begin{proof}[Proof of Theorem~\ref{thm:mo-quadratic}]
Use the notation fixed above, and choose any fixed ordering of the $d$ prepared states
\[
|\psi_{a,b}^{R}\rangle,
\qquad
|\psi_{a,b}^{I}\rangle
\qquad ((a,b)\in[r]\times[s]).
\]
Denote the resulting list by $(|\psi_j\rangle)_{j=1}^{d}$.
For each sign vector
\[
\eta=\bigl(\eta_{a,b}^{R},\eta_{a,b}^{I}\bigr)_{(a,b)\in[r]\times[s]}\in\{\pm1\}^{d},
\]
write $\eta_j$ for the corresponding sign in the same ordering, and define the complex $r\times s$ matrix
\[
(X_{\eta})_{a,b}=\eta_{a,b}^{R}-i\eta_{a,b}^{I}
\qquad ((a,b)\in[r]\times[s]),
\]
and the skew-Hermitian generator
\[
K_{\eta}=\begin{pmatrix}
0 & X_{\eta}\\
- X_{\eta}^{\dagger} & 0
\end{pmatrix}.
\]
Set $t=1/(4n^2)$ and write $U_{\eta}=\exp(tK_{\eta})$.

For $u\ge 0$, let $G_{\eta}(u)=e^{-uK_{\eta}}P_0e^{uK_{\eta}}$.
Taylor expansion with integral remainder, applied to the operator-valued curve $u\mapsto G_{\eta}(u)$ and with $\mathrm{ad}_{K_{\eta}}(X)=[X,K_{\eta}]$, gives
\[
G_{\eta}(t)=P_0+t[P_0,K_{\eta}]+\frac{t^2}{2}[[P_0,K_{\eta}],K_{\eta}]+R_{\eta}(t),
\]
where
\[
R_{\eta}(t)=\int_0^t \frac{(t-u)^2}{2}\,\mathrm{ad}_{K_{\eta}}^3\bigl(G_{\eta}(u)\bigr)\,du.
\]
Since $\|[A,K_{\eta}]\|\le 2\|A\|\,\|K_{\eta}\|$ for the operator norm, one has $\|\mathrm{ad}_{K_{\eta}}\|\le 2\|K_{\eta}\|$ on $\mathcal{L}(\mathcal{H})$.
Moreover, $\|G_{\eta}(u)\|=\|P_0\|=1$ for all $u$, and therefore
\[
\|R_{\eta}(t)\|\le \frac{4}{3}t^3\|K_{\eta}\|^3.
\]
Every entry of $X_{\eta}$ has modulus $\sqrt{2}$.
Also,
\[
K_{\eta}^{\dagger}K_{\eta}=\begin{pmatrix} X_{\eta}X_{\eta}^{\dagger} & 0\\ 0 & X_{\eta}^{\dagger}X_{\eta}\end{pmatrix},
\]
so $\|K_{\eta}\|=\|X_{\eta}\|\le \|X_{\eta}\|_{\mathrm{F}}=\sqrt{2rs}\le n/\sqrt{2}$.
In particular,
\[
\|K_{\eta}\|\le n,
\qquad
\|R_{\eta}(t)\|\le \frac{1}{48n^3}.
\]
This bound is uniform in $\eta$ and $n$, and its lack of sharpness is harmless because the linear term is of order $n^{-2}$ whereas the remainder is only of order $n^{-3}$.

The first two commutators have the explicit block form
\[
[P_0,K_{\eta}]=\begin{pmatrix} 0 & X_{\eta}\\ X_{\eta}^\dagger & 0\end{pmatrix},
\qquad
[[P_0,K_{\eta}],K_{\eta}]=\begin{pmatrix} -2X_{\eta}X_{\eta}^\dagger & 0\\ 0 & 2X_{\eta}^\dagger X_{\eta}\end{pmatrix}.
\]
Evaluating the linear term on the prepared states yields
\[
\langle \psi_{a,b}^{R}|[P_0,K_{\eta}]|\psi_{a,b}^{R}\rangle=\mathrm{Re}\bigl((X_{\eta})_{a,b}\bigr)=\eta_{a,b}^{R},
\]
\[
\langle \psi_{a,b}^{I}|[P_0,K_{\eta}]|\psi_{a,b}^{I}\rangle=-\mathrm{Im}\bigl((X_{\eta})_{a,b}\bigr)=\eta_{a,b}^{I}.
\]
For the quadratic term, the constant modulus $|(X_{\eta})_{a,b}|^2=2$ implies
\[
(X_{\eta}X_{\eta}^{\dagger})_{a,a}=2s,
\qquad
(X_{\eta}^{\dagger}X_{\eta})_{b,b}=2r,
\]
so for both state types,
\[
\left\langle \psi\,\middle|\,\frac{1}{2}[[P_0,K_{\eta}],K_{\eta}]\,\middle|\,\psi\right\rangle=r-s,
\qquad
\psi\in\{|\psi_{a,b}^{R}\rangle,|\psi_{a,b}^{I}\rangle\}.
\]
Thus the second-order contribution is a uniform constant shift, independent of both $\eta$ and the tested state.

Now construct an MO-1QFA $Q_n$ with any fixed initial unit state $|\varphi_0\rangle\in\mathcal{H}$ and accepting projector $P_{\mathrm{acc}}=P_0$.
Its alphabet contains prepare symbols $p_j$ and test symbols $\tau_{\eta}$.
Each $p_j$ applies a unitary $V_j$ satisfying $V_j|\varphi_0\rangle=|\psi_j\rangle$, and each $\tau_{\eta}$ applies $U_{\eta}$.
For the length-$2$ word $p_j\tau_{\eta}$,
\[
f_{Q_n}(p_j\tau_{\eta})=\langle \psi_j|U_{\eta}^{\dagger}P_0U_{\eta}|\psi_j\rangle
=\frac12+t\eta_j+(r-s)t^2+\delta_{j,\eta},
\qquad
|\delta_{j,\eta}|\le \frac{1}{48n^3}.
\]
Because $|r-s|\le 1$ and $t=1/(4n^2)$,
\[
t-|r-s|t^2-\frac{1}{48n^3}
\ge t-\frac{t}{16}-\frac{t}{24}
=\frac{43}{48}t>0.
\]
Therefore,
\[
f_{Q_n}(p_j\tau_{\eta})>\frac12
\quad\Longleftrightarrow\quad
\eta_j=+1.
\]
The family of tests $\{\tau_{\eta}\}_{\eta\in\{\pm1\}^{d}}$ thus shatters the $d$ prepared prefixes $\{p_j\}_{j=1}^{d}$.

Now let $P$ be an $m$-state one-way PFA recognizing the same language under some strict cutpoint $\mu$.
For each $j$, let $\delta_j\in\Delta_{m-1}$ be the state distribution of $P$ after reading $p_j$.
For each $\eta$, let $b_{\eta}\in[0,1]^m$ be the conditional acceptance vector after reading $\tau_{\eta}$ and the end-marker.
Then
\[
f_P(p_j\tau_{\eta})=\delta_j^\top b_{\eta}.
\]
Since $P$ recognizes the same language, the affine halfspaces
\[
\{x\in\Delta_{m-1}: x^\top b_{\eta}>\mu\}
\qquad (\eta\in\{\pm1\}^{d})
\]
shatter the points $\delta_1,\dots,\delta_d$.
By Proposition~\ref{prop:simplex-vc}, this forces $d\le m$.
Hence $m\ge d=\lfloor n^2/2\rfloor$.
\end{proof}

\begin{proof}[Proof of Corollary~\ref{cor:mo-worstcase}]
The upper bound follows by viewing an MO-1QFA as a special case of a mixed-state one-way general quantum finite automaton on an $n$-dimensional Hilbert space.
The mixed-state linearization from~\cite{ChenWuQuadratic} yields an $n^2$-state GFA, which Proposition~\ref{prop:gfa-to-pfa} converts over the corresponding ordered coefficient field to a $K$-PFA with at most $2n^2+6$ states.
The lower bound is Theorem~\ref{thm:mo-quadratic}.
\end{proof}

\begin{remark}
\label{rem:pure-vs-mixed}
The pure-state restriction does not reduce the worst-case order of exact probabilistic simulation.
Even for MO-1QFA, one still needs a PFA of quadratic size in the quantum dimension, matching the order already seen for 1gQFA and 1QCFA\@.
What changes is the geometric source of the lower bound.
For mixed-state models, the quadratic obstruction comes from the ambient affine state space of density operators itself.
For MO-1QFA, by contrast, the reachable pure-state set has affine dimension only $O(n)$; the missing quadratic freedom is recovered from the local geometry of the measurement orbit $\mathrm{Gr}(\lfloor n/2\rfloor,n)$, whose tangent directions support $\lfloor n^2/2\rfloor$ independent threshold directions.
\end{remark}

The two subsections above establish the exact one-way simulation costs by the same prepare--test principle, but the geometric source of the quadratic dimension differs in a revealing way.  For 1QCFA, the prepared classical--quantum configurations themselves fill the block-diagonal operator space, so the lower bound comes from the state geometry.  For MO-1QFA, the prepared pure states lie on a manifold of only linear dimension; the missing quadratic freedom is supplied by the accepting projector, whose unitary orbit provides $\lfloor n^2/2\rfloor$ independent test directions.  In both cases, the common output is a family of finite prefix--suffix tests that shatters a set of prepared objects.  The next section extracts this finite shadow through sign-rank and calibrates it against Forster's spectral method.

\section{Forster's spectral method and finite sign-rank witnesses}
\label{sec:signrank}

The prepare--test constructions of Section~\ref{sec:prepare-test} are geometric and continuous, yet their obstruction to probabilistic simulation is already captured by finite combinatorial data.  Once finite sets of prefixes and suffixes are fixed, a strict-cutpoint automaton determines a \emph{sign matrix} whose entries record which concatenations lie above the cutpoint.  Sign-rank is the natural invariant of this matrix: every $m$-state probabilistic simulator gives a rank-$m$ real realization of the same sign pattern after the cutpoint is subtracted, so a high sign-rank forces a large simulator.

Forster's spectral method is a standard technique for proving sign-rank lower bounds on explicit square sign matrices~\cite{Forster2002}.  In its basic form, it asserts that every matrix $S\in\{\pm1\}^{L\times L}$ satisfies
\[
\mathrm{rank}_{\pm}(S)\ge \frac{L}{\|S\|_2},
\]
where $\|S\|_2$ is the spectral norm.  The certificate is especially effective when a square restriction has unusually small spectral norm---as in Hadamard-type patterns---because the lower bound is then obtained from a directly computable quantity.

The finite witness produced by prepare--test, however, has a different shape.  It is the \emph{complete shattering matrix}: a $d\times 2^d$ rectangular matrix whose columns realize every sign pattern on a $d$-element prepared set.  Its sign-rank is exactly $d$, and this exact value uses the full family of tests rather than a single square spectral restriction.  The two approaches are therefore complementary: the square Forster certificate gives a useful analytic lower-bound tool, while the complete shattering geometry achieves the quadratic lower bound at its natural dimension.  The first subsection records the finite sign-rank witness; the second subsection calibrates the scale at which the basic square spectral method operates.

\subsection{Finite sign-rank witnesses from prepare--test}
\label{subsec:signrank-general}

Fix finite sets $X,Y\subseteq\Sigma^\ast$.
Given a machine $A$ and a strict cutpoint $\lambda$, define the induced sign matrix
\[
S^{A,\lambda}_{x,y}=\mathrm{sgn}(f_A(xy)-\lambda)\in\{\pm1\},
\qquad x\in X,\ y\in Y,
\]
where $\mathrm{sgn}(t)=+1$ for $t>0$ and $\mathrm{sgn}(t)=-1$ for $t\le 0$.

\begin{theorem}
\label{thm:signrank-mech}
Let $A$ recognize a language $L$ with strict cutpoint $\lambda$.
Assume that an $m$-state one-way PFA $P$ recognizes the same language $L$ with strict cutpoint $\mu$.
Then, for every finite $X,Y\subseteq\Sigma^\ast$,
\[
m\ge \mathrm{rank}_{\pm}(S^{A,\lambda}).
\]
\end{theorem}

\begin{proof}
For each $x\in X$, let $a_x\in\Delta_{m-1}$ be the state distribution of $P$ after reading $x$.
For each $y\in Y$, let $b_y\in[0,1]^m$ be the column vector of conditional acceptance probabilities after reading $y$ and the end-marker.
Then
\[
f_P(xy)=a_x^\top b_y.
\]
Since $X\times Y$ is finite, choose $\varepsilon>0$ smaller than every positive gap $f_P(xy)-\mu$ occurring on this finite set.  If there is no positive gap, choose any $\varepsilon>0$.
Put $\mu'=\mu+\varepsilon$ and
\[
c_y=b_y-\mu'{\bf 1}\in\mathbb{R}^m,
\]
where ${\bf 1}$ is the all-one vector.
The real matrix $R$ with entries
\[
R_{x,y}=a_x^\top c_y=a_x^\top b_y-\mu'
\]
has rank at most $m$.
Because each $a_x$ is a probability vector, the cutpoint subtraction is absorbed into the same $m$ coordinates rather than requiring an additional affine coordinate.
Moreover, by the choice of $\varepsilon$, the shift from $\mu$ to $\mu'$ preserves the strict signs on the finite set $X\times Y$:
\[
R_{x,y}>0
\iff
f_P(xy)>\mu
\iff
xy\in L
\iff
f_A(xy)>\lambda.
\]
Thus $S^{A,\lambda}_{x,y}R_{x,y}>0$ for every $(x,y)$, and therefore $\mathrm{rank}_{\pm}(S^{A,\lambda})\le m$.
\end{proof}

The finite matrices arising from the prepare--test proofs have a particularly simple universal form.  For $d\ge 1$, define the complete shattering matrix
\[
\mathsf{C}_d\in\{\pm1\}^{[d]\times\{\pm1\}^d},
\qquad
(\mathsf{C}_d)_{j,\eta}=\eta_j.
\]
Its columns list all sign patterns on $d$ prepared objects.

\begin{lemma}
\label{lem:complete-shattering-signrank}
For every $d\ge 1$,
\[
\mathrm{rank}_{\pm}(\mathsf{C}_d)=d.
\]
\end{lemma}

\begin{proof}
The upper bound $\mathrm{rank}_{\pm}(\mathsf{C}_d)\le d$ follows from $\mathsf{C}_d$ itself, whose row rank is $d$.  For the reverse inequality, let $R$ be any real matrix with the same sign pattern as $\mathsf{C}_d$, and let $W\subseteq\mathbb{R}^d$ be the column space of $R$.
For every $\eta\in\{\pm1\}^d$, the column $R_{\ast,\eta}$ lies in the open orthant
\[
O_\eta=\{z\in\mathbb{R}^d:\eta_j z_j>0\text{ for all }j\in[d]\}.
\]
Hence $W$ intersects every open orthant.
If $W$ were a proper subspace, choose a nonzero vector $v\in W^\perp$.
Choose $\eta_j=+1$ when $v_j=0$ and $\eta_j=\mathrm{sgn}(v_j)$ when $v_j\ne0$.
Then $v^\top z>0$ for every $z\in O_\eta$, contradicting $z\in W\subseteq v^\perp$.
Therefore $W=\mathbb{R}^d$, so $\mathrm{rank}(R) = d$.
\end{proof}

The lower-bound constructions in Section~\ref{sec:prepare-test} realize exactly these complete shattering matrices.
Thus the VC-dimension proofs there can be repackaged as sign-rank lower bounds without changing the automata or the cutpoint.

\begin{proposition}
\label{prop:section3-signrank-form}
The two lower-bound constructions of Section~\ref{sec:prepare-test} have the following finite sign-rank form.
\begin{enumerate}[label=(\alph*),nosep]
\item For $c\ge2$ and $q\ge2$, the $(c,q)$-1QCFA constructed in the proof of Theorem~\ref{thm:hybrid-main}(b) has finite prefix and suffix sets whose induced sign matrix is $\mathsf{C}_{cq^2-1}$.  Consequently every equivalent one-way PFA has at least $cq^2-1$ states.
\item For $n\ge2$, the $n$-MO-1QFA constructed in the proof of Theorem~\ref{thm:mo-quadratic} has finite prefix and suffix sets whose induced sign matrix is $\mathsf{C}_{\lfloor n^2/2\rfloor}$.  Consequently every equivalent one-way PFA has at least $\lfloor n^2/2\rfloor$ states.
\end{enumerate}
\end{proposition}

\begin{proof}
For the 1QCFA construction, put $d=cq^2-1$ and use the prefix set
\[
X=\{p_1,\dots,p_d\}
\]
and the suffix set
\[
Y=\{\tau_\eta:\eta\in\{\pm1\}^d\}.
\]
The computation in the proof of Theorem~\ref{thm:hybrid-main}(b) gives
\[
f_{\mathcal{A}_{c,q}}(p_\ell\tau_\eta)=\frac12+t\eta_\ell.
\]
Hence $S^{\mathcal{A}_{c,q},1/2}_{p_\ell,\tau_\eta}=\eta_\ell$, so the induced sign matrix is $\mathsf{C}_d$.
Lemma~\ref{lem:complete-shattering-signrank} and Theorem~\ref{thm:signrank-mech} give the lower bound $m\ge d$ for every equivalent one-way PFA.

For the MO-1QFA construction, put $d=\lfloor n^2/2\rfloor$ and use the prefix set
\[
X=\{p_1,\dots,p_d\}
\]
and the suffix set
\[
Y=\{\tau_\eta:\eta\in\{\pm1\}^d\}.
\]
The proof of Theorem~\ref{thm:mo-quadratic} establishes
\[
f_{Q_n}(p_j\tau_\eta)>\frac12
\quad\Longleftrightarrow\quad
\eta_j=+1.
\]
Therefore $S^{Q_n,1/2}_{p_j,\tau_\eta}=\eta_j$, and the induced sign matrix is $\mathsf{C}_d$.
Again Lemma~\ref{lem:complete-shattering-signrank} and Theorem~\ref{thm:signrank-mech} yield $m\ge d$.
\end{proof}

This proposition is the finite form of the prepare--test mechanism.  The geometry in Section~\ref{sec:prepare-test} supplies the automata that realize all sign patterns; sign-rank records the resulting obstruction in a prefix--suffix matrix.  The comparison with Forster's method now becomes transparent: the complete shattering matrix gives an exact finite witness, while the square spectral method asks for a square restriction with small spectral norm.

\subsection{What the square spectral method can certify}
\label{subsec:spectral-barrier}

Forster's inequality converts small spectral norm into large sign-rank.  In its basic square form, however, there is an intrinsic scale limitation: every $L\times L$ sign matrix has Frobenius norm $L$, hence spectral norm at least $\sqrt{L}$, so the certified lower bound $L/\|S\|_2$ never exceeds $\sqrt{L}$.  To certify a quadratic lower bound $\Omega(n^2)$ by a single square spectral inequality, one must therefore start with a square restriction of side length $L=\Omega(n^4)$.  The following two statements make this scale comparison precise.

\begin{proposition}
\label{prop:acceptance-rank}
Let $Q$ be a one-way quantum finite automaton on an $n$-dimensional quantum register, either pure-state or mixed-state, and let $X,Y$ be finite sets of size $L$.
Then,
\[
\mathrm{rank}_{\pm}(S^{Q,\lambda})\le n^2.
\]
\end{proposition}

\begin{proof}
For each $x\in X$, let $\rho_x$ be the state after reading $x$.
For each $y\in Y$, let $E_y$ be the acceptance effect corresponding to continuing with $y$ and the end-marker.
Then
\[
f_Q(xy)=\mathrm{Tr}(E_y\rho_x)=\langle \mathrm{vec}(\rho_x),\mathrm{vec}(E_y)\rangle,
\]
so the matrix $(f_Q(xy))_{x,y}$ factors through $\mathbb{R}^{n^2}$ and therefore has rank at most $n^2$.
Since every $\rho_x$ has trace one, subtracting the cutpoint can be absorbed into the effect side as $\mathrm{Tr}((E_y-\lambda I)\rho_x)$. To obtain a strict sign-rank realization, choose $\varepsilon>0$ smaller than every positive gap $f_Q(xy)-\lambda$ on $X\times Y$, or choose any $\varepsilon>0$ if no positive gap occurs.  The shifted matrix
\[
M^{\varepsilon}_{x,y}=f_Q(xy)-(\lambda+\varepsilon)
=\mathrm{Tr}((E_y-(\lambda+\varepsilon)I)\rho_x)
\]
has rank at most $n^2$ and satisfies $S^{Q,\lambda}_{x,y}M^{\varepsilon}_{x,y}>0$ for every $(x,y)$.  Hence $\mathrm{rank}_{\pm}(S^{Q,\lambda})\le n^2$.
\end{proof}

\begin{theorem}
\label{thm:spectral-barrier}
Let $S\in\{\pm1\}^{L\times L}$ be any sign matrix.
The numerical certificate supplied by the basic square Forster inequality
\[
\mathrm{rank}_{\pm}(S)\ge \frac{L}{\|S\|_2}
\]
is at most $\sqrt{L}$.
Consequently, any attempt to prove an explicit $\Omega(n^2)$ lower bound for an $n$-dimensional one-way quantum automaton by combining a concrete square prefix--suffix construction with this square spectral certificate must use $L=\Omega(n^4)$.
\end{theorem}

\begin{proof}
Since $S$ has $L^2$ entries of magnitude $1$, its Frobenius norm is $\|S\|_{\mathrm{F}}=L$.
Because $\|S\|_2\ge \|S\|_{\mathrm{F}}/\sqrt{L}=\sqrt{L}$, the numerical lower bound certified by Forster's inequality satisfies
\[
\frac{L}{\|S\|_2}\le \sqrt{L}.
\]
Therefore a lower bound of order $n^2$ through this square spectral route requires $\sqrt{L}=\Omega(n^2)$, namely $L=\Omega(n^4)$.
\end{proof}

The communication-complexity interpretation gives the same distinction in another language.
Given finite prefix and suffix sets $X,Y$, define the Boolean function
\[
g_{X,Y}(x,y)=1 \iff xy\in L.
\]
By the classical characterization of Paturi and Simon, the unbounded-error communication complexity of $g_{X,Y}$ is $\Theta(\log \mathrm{rank}_{\pm}(S^{A,\lambda}))$~\cite{PaturiSimon1986}.
Thus strict-cutpoint lower bounds on finite prefix--suffix restrictions become unbounded-error communication lower bounds.
Related parameters such as margin and factorization norms refine this viewpoint and connect spectral norm, sign-rank, and communication complexity more broadly~\cite{LinialMendelsonSchechtmanShraibman2007}.

\begin{remark}
\label{rem:explicit-gap}
Theorem~\ref{thm:mo-quadratic} proves an $\Omega(n^2)$ lower bound for MO-1QFA, matching the quadratic order of the $2n^2+6$ upper bound in Corollary~\ref{cor:mo-worstcase}.
Proposition~\ref{prop:section3-signrank-form} shows that this lower bound has an exact finite sign-rank witness.  Theorem~\ref{thm:spectral-barrier} explains the precise limitation of a proof that uses only the basic square form of Forster's spectral method.  This is a limitation of the certificate, rather than of sign-rank itself: the complete shattering matrix already gives the required quadratic-order witness at its natural dimension.
\end{remark}

\section{Quantum advantage and the boundary of simulation}
\label{sec:twoway-boundary}

The preceding sections establish tight simulation costs for two one-way quantum models and record their finite sign-rank witnesses.  Together with known two-way separations, these results give a clean picture of how quantum advantage changes across finite automata.  For the effective one-way models considered in this paper, recognition under strict cutpoints is already shared with probabilistic automata, and the comparison is therefore the exact state cost of simulation.  Bounded-error two-way motion shifts the comparison to recognition power and running time.

Within the one-way regime, state succinctness remains relevant.  Measure-once one-way quantum finite automata have limited recognition power under isolated cutpoints, recognizing only group languages, a proper subclass of regular languages~\cite{BrodskyPippenger2002}.  Within this restricted class, modular languages already give exponential state succinctness: for prime moduli, one-way quantum automata use $O(\log p)$ states, whereas probabilistic automata require $\Omega(p)$ states~\cite{AmbainisFreivalds1998}.  Dense-coding methods and quantum--classical constructions have extended this descriptional-complexity theme further~\cite{AmbainisNayakTaShmaVazirani2002,QiuLiMateusSernadas2015}.  More general one-way models recover exactly the regular languages with bounded error~\cite{LiQiuZouEtAl2012}, and hybrid one-way automata fit the same succinctness picture~\cite{ZhengQiuGruska2013,QiuLiMateusSernadas2015}.  These results set the stage for the invariant studied in Sections~\ref{sec:prepare-test}--\ref{sec:signrank}: once probabilistic recognition is assured, quantum advantage persists as the state cost of exact threshold simulation, and the sharp quadratic laws established above close this question for the two representative one-way models.

Two-way head motion gives the next boundary.  Ambainis and Watrous showed that 2QCFA recognize the palindrome language with bounded error, a language outside bounded-error 2PFA recognition, and that 2QCFA recognize $L_{eq}=\{a^n b^n:n\ge 0\}$ in expected polynomial time whereas 2PFA require exponential expected time~\cite{AmbainisWatrous,DworkStockmeyer1990}.  These separations show that once two-way quantum--classical control is allowed, the first obstruction to classical simulation may occur before an exact state-count comparison: probabilistic recognition or efficient probabilistic simulation may already be unavailable.

These two-way results are included only to locate the scope of the one-way theorems.  Sections~\ref{sec:prepare-test}--\ref{sec:signrank} prove exact state-cost laws in a regime where recognition is already shared.  Once the model changes enough to alter recognition power or efficient simulation, the first obstruction is no longer the number of simulator states.  The role of the one-way results is therefore precise: they identify the finite-dimensional geometry that remains visible after recognition-level differences have been neutralized by the strict-cutpoint setting.

\section{Conclusion}
\label{sec:discussion}

This paper identifies exact probabilistic simulation cost as the natural quantitative measure of quantum advantage for finite automata under strict cutpoints, and establishes sharp laws for two representative one-way models.  A $(c,q)$-1QCFA has exact simulation cost $\Theta(cq^2)$, governed by its block-diagonal classical--quantum operator space; an $n$-dimensional MO-1QFA has worst-case cost $\Theta(n^2)$, driven by the quadratic-dimensional orbit of its accepting projector.  The common prepare--test mechanism turns the real degrees of freedom visible to the accepting functional into threshold shattering, and the sign-rank viewpoint records the same obstruction in finite prefix--suffix matrices while clarifying the complementary role of Forster's spectral method.

These one-way simulation laws occupy a specific position in the landscape of quantum advantage.  Recognition power, state succinctness, and simulation cost form a hierarchy of progressively finer comparisons between quantum and probabilistic finite automata.  In the effective one-way strict-cutpoint regime, the recognition-level separation disappears because both models recognize the same stochastic languages; the remaining advantage is quantitative, and this paper sharpens it to exact simulation cost.  Two-way motion introduces a stronger boundary where simulation can fail for recognition or time-complexity reasons.  The quadratic laws proved here therefore characterize the effective one-way regime as the domain where quantum advantage survives purely as exact classical simulation cost.

\end{document}